%% file: arXiv.tex
\documentclass[aps,nopacs,11pt]{revtex4}

\include{pretex}

\begin{document}
\title{\Large Non-asymptotic Entanglement Distillation}
\author{Kun Fang$^{1,2}$}
\email{kf383@cam.ac.uk}
\author{Xin Wang$^{1,3}$}
\email{xwang93@umd.edu}
\author{Marco Tomamichel$^{1}$}
\email{marco.tomamichel@uts.edu.au}
\author{Runyao Duan$^{4,1}$}
\email{duanrunyao@baidu.com}

\affiliation{$^1$Centre for Quantum Software and Information, Faculty of Engineering and Information Technology, University of Technology Sydney, NSW 2007, Australia}
\affiliation{$^2$ Department of Applied Mathematics and Theoretical Physics,\\ University of Cambridge, Cambridge, CB3 0WA, UK}

\affiliation{$^3$Joint Center for Quantum Information and Computer Science, University of Maryland, College Park, Maryland 20742, USA}
\affiliation{$^4$Institute for Quantum Computing, Baidu Research, Beijing 100193, China}

\thanks{A preliminary version of this paper was accepted as a long talk presentation at the 17th Asian Quantum Information Science Conference (AQIS 2017).}

\begin{abstract}
 Entanglement distillation, an essential quantum information processing task,  refers to the conversion from multiple copies of noisy entangled states to a smaller number of highly entangled states. In this work, we study the non-asymptotic fundamental limits for entanglement distillation. We investigate the optimal tradeoff between the distillation rate, the number of prepared states, and the error tolerance. First, we derive the one-shot distillable entanglement under completely positive partial transpose preserving operations as a semidefinite program and demonstrate an exact characterization via the quantum hypothesis testing relative entropy. Second, we establish efficiently computable second-order estimations of the distillation rate for general quantum states. In particular, we provide explicit as well as approximate evaluations for various quantum states of practical interest, including pure states, mixture of Bell states, maximally correlated states and isotropic states. 

\end{abstract}

\maketitle

\section{Introduction}

Quantum entanglement is a striking feature of quantum mechanics and a key ingredient in many quantum information processing tasks~\cite{Horodecki2009a}, including teleportation~\cite{Bennett1993}, superdense coding~\cite{Bennett1992}, and quantum cryptography~\cite{Bennett1984,Bennett1992,Ekert1991}. All these protocols necessarily rely on entanglement resources. It is thus of great importance to transform less entangled states into more suitable ones such as maximally entangled states $\Psi_k$. This procedure is known as \textit{entanglement distillation } or \textit{entanglement concentration}~\cite{Bennett1996c}.


The task of entanglement distillation allows two parties (Alice and Bob) to perform a set of free operations $\O$, for example, local operations and classical communication (LOCC). The \textit{distillable entanglement} characterizes the optimal rate at which one can asymptotically obtain maximally entangled states from a collection of identically and independently distributed (i.i.d) prepared entangled states~\cite{Bennett1996, Rains1999a,Plenio2007}. The concise definition of distillable entanglement by the class of operation $\Omega$ can be given by
\begin{align}
E_{\text{D},\Omega}(\rho_{AB}) :=\sup \Big\{\,r \, \Big|\, \lim_{n \to \infty} \Big(\inf_{\Pi\in \Omega}  \|\Pi(\rho_{AB}^{\ox n})- \Psi_{2}^{\ox r\cdot n}\|_1\Big)=0 \,\Big\}.
\end{align}
Entanglement distillation from non-i.i.d prepared states has also been considered recently in~\cite{Waeldchen2015}.
Distillable entanglement is a fundamental entanglement measure which captures the resource character of quantum entanglement. Up to now, it remains unknown how to compute the distillable entanglement for general quantum states and various approaches \cite{Popescu1997,Vedral1998,Rains1999,Vidal2002,Rains2001,Horodecki2000a,Christandl2004,Leditzky2017,Wang2016} have been developed to evaluate this important quantity. In particular, the Rains bound \cite{Rains2001} as well as the squashed entanglement~\cite{Christandl2004} are arguably the best general upper bounds while the hashing bound \cite{Devetak2005d} is the best lower bound for the distillable entanglement.


The conventional approach to studying entanglement distillation is to consider the asymptotic limit (first-order asymptotics), assuming our access to an unbounded number of i.i.d copies of a quantum state. In a realistic setting, however, the resources are finite and the number of i.i.d prepared states is necessarily limited. More importantly, it is very difficult to perform coherent state manipulations over large numbers
of systems. Therefore, it becomes crucial to characterize how well we can faithfully distill maximally entangled states from a finite number of copies of the prepared states.
Since the first-order asymptotics are insufficient to give a precise estimation when $n$ is finite, it is necessary to consider higher order asymptotics. Specifically, we consider estimating the optimal distillation rate $r_n$ for $n$ copies of the state to the order $\sqrt{n}$, for example, 
$r_n = a + \frac{b}{\sqrt{n}} + O\big(\frac{\log n}{n}\big)$, where $a$ and $b$ are the so-called first and second-order asymptotics respectively. The first-order term $a$ determines the asymptotic rate of $r_n$ while the second-order term $b$ indicates how fast the rate $r_n$ converges to $a$.
The second-order estimation is especially accurate for large blocklength $n$ where the higher-order term $O\big(\frac{\log n}{n}\big)$ is negligible. In particular, for practical use, it is desirable to find efficiently computable coefficients $a,b$.

The study of such non-asymptotic scenarios has recently garnered great interest in classical information theory (e.g., \cite{Polyanskiy2010,Hayashi2009,Tan2014}) as well as in quantum information theory (e.g., \cite{Wang2012,Renes2011,Berta2011,Tomamichel2013a,Leung2015c,Datta2013c,Matthews2014,Beigi2015,Tomamichel2015b,Tomamichel2016,Wang2016g,Cheng2017b,Chubb2017}).
Here we study the setting of entanglement distillation.  A non-asymptotic analysis of entanglement distillation will help us better exploit the power of entanglement in a realistic setting.  
Previously, the one-shot distillable entanglement was studied in~\cite{Buscemi2010b,Brandao2011a}. But these bounds are not known to be computable in general, which makes it difficult to apply them as experimental benchmarks. These one-shot bounds are not suitable to establish second-order estimations either. Datta and Leditzky studied the second-order estimation of distillable entanglement under LOCC operations for pure states \cite{Datta2015a}. Here, we go beyond their results by considering more general operations and states.


The remainder of this paper is organized as follows. In Section~\ref{One-shot entanglement distillation} we study the one-shot entanglement distillation under completely positive partial transpose preserving operations and find that the one-shot rate is efficiently computable via a semidefinite program (SDP). Based on this SDP, we present an exact characterization of the one-shot rate via the quantum hypothesis testing relative entropy, which can be seen as a one-shot analog of the Rains bound. In Section~\ref{Non-asymptotic entanglement distillation} we investigate the entanglement distillation for $n$-fold tensor product states and provide second-order estimations for general quantum states. In Section~\ref{sec:examples} we apply our second-order estimations to various quantum states of practical interest, including pure states, mixture of Bell states, maximally correlated states and isotropic states.

\section{Preliminaries}



In the following, we will frequently use symbols such as $A$ (or $A'$) and $B$ (or $B'$) to denote finite-dimensional Hilbert spaces associated with Alice and Bob, respectively. 
A quantum state on system $A$ is a positive operator $\rho_A$ with unit trace. The set of quantum states is denoted as $S(A):= \{\,\rho_A \geq 0 \,|\, \tr \rho_A = 1\,\}$. The set of subnormalized quantum states is denoted as $S_\leq(A):= \{\,\rho_A \geq 0 \,|\, 0 < \tr \rho_A \leq 1\,\}$.
We call a positive operator separable if it can be written as a convex combination of tensor product positive operators.
 

A quantum operation is characterized by a completely positive and trace-preserving (CPTP) linear map.
There are several different classes of quantum operations we often use. We call a bipartite quantum operation LOCC if it can be realized by local operations and classical communication. If only one-way classical communication is allowed, say, classical information from Alice to Bob, we call it 1-LOCC. While LOCC, including 1-LOCC, emerges as the natural class of operations in many important quantum information tasks, its mathematical structure is complex and difficult to characterize~\cite{Chitambar2014}. Therefore we may consider larger but mathematically more tractable classes of operations.  
The operations most frequently employed beyond
LOCC are the so-called completely positive partial transpose preserving (PPT) operations and separable (SEP) operations. 
A bipartite quantum operation $\Pi_{AB \to A'B'}$ is said to be a PPT (or SEP) operation if its  Choi-Jamio\l{}kowski matrix $J_\Pi = \sum_{i,j,m,k} \ket{i_Aj_B}\bra{m_Ak_B} \ox \Pi(\ket{i_Aj_B}\bra{m_Ak_B})$ is positive under partial transpose (or separable) across the bipartition of $AA':BB'$, where $\{\ket{i_A}\}$ and $\{\ket{j_B}\}$ are orthonormal bases on Hilbert spaces $A$ and $B$, respectively. In particular, PPT operations can be characterized via semidefinite conditions~\cite{Rains2001}. The entanglement theory under PPT operations has been studied in the literature~(e.g.,~\cite{Audenaert2003,Wang2016d,Matthews2008,Wang2018b}) and offers the limitations of LOCC.

A well known fact is that the classes of above introduced operations obey the
strict inclusions~\cite{Bennett1999b,Horodecki2009a},
$\text{1-LOCC} \subsetneq \text{LOCC} \subsetneq  \text{SEP} \subsetneq \text{PPT}$.
As a consequence, for any quantum state $\rho_{AB}$ we have the following chain of inequalities,
\begin{align}
    E_{\DLLOCC}(\rho_{AB}) \leq E_{\DLOCC}(\rho_{AB}) \leq E_{\DSEP}(\rho_{AB}) \leq E_{\DPPT}(\rho_{AB}).
\end{align}
This allows us to use $E_{\DLLOCC}$ and $E_{\DPPT}$ as lower and upper bounds, respectively, for entanglement distillation under LOCC operations.


Quantum hypothesis testing is the task of distinguishing two possible states of a system, $\rho_0$ and $\rho_1$. Two hypotheses are studied: the null hypothesis $H_0$ is that the state is $\rho_0$; the alternative hypothesis $H_1$ is that the state is $\rho_1$. We are allowed to perform a measurement presented by the POVM $\{M, \1-M\}$ with corresponding classical outcomes $0$ and~$1$. If the outcome is $0$, we accept the null hypothesis. Otherwise, we accept the alternative one. The probabilities of \emph{type-I} and \emph{type-II} error are given by $\tr (\1-M) \rho_0$ and $\tr M \rho_1$, respectively. Quantum hypothesis testing relative entropy considers minimizing the type-II error while keeping type-I error within a given error tolerance. Specifically, it is defined as 
\begin{align}
\label{hypothesis testing definition}
D_H^\ve(\rho_0\|\rho_1) :=-\log\min \big\{ \tr M\rho_1 \,\big|\, 0\le M\le \1,\, 1-\tr M\rho_0\le\ve \big\}.
\end{align}
 Throughout the paper we take the logarithm to be base two unless stated otherwise. Note that $D_H^\ve$ is a fundamental quantity in quantum theory \cite{Helstrom1976,Hiai1991,Ogawa2000,Hayashi2017g} and can be solved by SDP---a powerful tool in quantum information theory with a plethora of applications (e.g., \cite{Doherty2002,Harrow2015,Berta2015a,Wang2016f,Kogias2015a,Li2017}).

For the convenience of the following discussion, we consider an extension of the quantum hypothesis testing relative entropy where the second argument $\rho_1$ is only restricted to be Hermitian (and not necessarily a positive semi-definite operator). We will also use the convention that $\log x = -\infty$ for $x \leq 0$ in case the optimal value $\tr M\rho_1 \leq 0$. Then Eq.~\eqref{hypothesis testing definition} is still a well-defined SDP. This extension is essential to obtain tight characterizations in this work as well as a related work on coherence distillation~\cite{Regula2017c}.


\section{One-shot entanglement distillation}
\label{One-shot entanglement distillation}

In this section, we consider distilling a maximally entangled state from a single copy of the resource state and study the tradeoff between the one-shot distillation rate and the fidelity of distillation. Since the distillation process cannot always be accomplished perfectly, we use the fidelity of distillation to characterize the performance of a given distillation task. Then the one-shot distillable entanglement
is defined as the logarithm of the maximal dimension of the maximally entangled state that we can obtain while keeping the infidelity of the distillation process within a given tolerance.

\begin{definition}
    For any bipartite quantum state $\rho_{AB}$, the fidelity of distillation under the operation class $\O$ is defined as~\cite{Rains2001}
    \begin{align}
        F_{\O}(\rho_{AB},k):=\max_{\Pi\in\O} \tr \Pi_{AB\to A'B'}(\rho_{AB})\Psi_k,
    \end{align}
    where $\Psi_k:=(1/k)\sum_{i,j=0}^{k-1} \ket{ii}\bra{jj}$ is the $k$-dimensional maximally entangled state and the maximization is taken over all possible operations $\Pi$ in the set $\O$.
\end{definition}

\begin{definition}
\label{EDPPT definition}
For any bipartite quantum state $\rho_{AB}$,
the one-shot $\ve$-error distillable entanglement under the operation class $\O$ is defined as
\begin{align}
\label{D PPT error}
E_{\text{\rm D},\O}^{(1),\ve}(\rho_{AB}) :=\log \max \Big\{\,k\in \mathbb{N}\, \Big| \, F_{\O}(\rho_{AB}, k)\ge 1-\ve\,\Big\}.
\end{align}
The asymptotic distillable entanglement is then given by the regularization:
\begin{align}
E_{\text{\rm D},\O}(\rho_{AB}) = \lim_{\ve\rightarrow 0} \lim_{n\rightarrow \infty} \frac{1}{n}E_{\text{\rm D},\O}^{(1),\ve}(\rho_{AB}^{\ox n}).
\end{align}
\end{definition}

Due to the linear characterization of PPT operations~\cite{Rains2001}, the one-shot distillable entanglement under PPT operations can be easily computed via the following optimization. The main ingredients of the proof use the symmetry of the maximally entangled state and the spectral decomposition of the swap operator, which are standard techniques used in literatures, e.g.,~\cite{Rains2001,Audenaert2003,Matthews2008}. We present the detailed proof here for the sake of completeness.

\begin{proposition} 
\label{EDPPT SDP characterization}
For any bipartite quantum state $\rho_{AB}$ and error tolerance $\ve$, the one-shot distillable entanglement under PPT operations is given by
\begin{align}\label{E 1 PPT}
   E_{\DPPT}^{(1),\ve}(\rho_{AB}) = \log \max \ & \lfloor1/\eta\rfloor \notag \\
    \text{\rm s.t.} \ & \, 0\le M_{AB} \le \1_{AB},\, \tr \rho_{AB}M_{AB}\ge 1-\ve,\,
   -\eta\1_{AB} \le M_{AB}^{T_{B}} \le  \eta\1_{AB} .
\end{align}
\end{proposition}
\begin{proof}
    From the definition of the one-shot distillable entanglement, we have 
    \begin{align}\label{Lemmma 1.3 tmp1}
    E_{\rm D,PPT}^{(1),\ve}(\rho) = \log \max \Big\{\,k\in \mathbb N \ \Big|\ \tr \Pi_{AB\to A'B}(\rho_{AB})\Psi_k \geq 1-\ve,\, \Pi\in \text{PPT}\,\Big\}.
    \end{align}
    According to the Choi-Jamio\l{}kowski representation of quantum operations~\cite{CHOI1975285,JAMIOLKOWSKI1972275}, we can represent the output state of operation $\Pi_{AB\rightarrow A'B'}$ via its \Choi matrix $J_{\Pi}$ as
    \begin{align}
    \Pi_{AB\to A'B'}(\rho_{AB}) = \tr_{AB} (J_{\Pi} \cdot \rho_{AB}^T \ox \1_{A'B'}).
    \end{align}
    By straightforward calculations, we have
    \begin{align}
      \tr \Pi_{AB\to A'B'}(\rho_{AB})\Psi_k & = \tr \big[\tr_{AB} (J_{\Pi} \cdot \rho_{AB}^T \ox \1_{A'B'})\big]\Psi_k\\
     &  = \tr J_{\Pi}\cdot (\rho_{AB}^T \ox \1_{A'B'}) (\1_{AB} \ox \Psi_k) \\
     & = \tr J_{\Pi}\cdot (\1_{AB} \ox \Psi_k) (\rho_{AB}^T \ox \1_{A'B'}) \\
     & = \tr \big[\tr_{A'B'} J_{\Pi}\cdot (\1_{AB} \ox \Psi_k)\big] \rho_{AB}^T. \label{Lemmma 1.3 tmp2}
    \end{align}
    Recall that $\Pi$ is a PPT operation if and only if its \Choi matrix $J_{\Pi}$ satisfies~\cite{Rains2001}
    \begin{align} \label{Lemmma 1.3 tmp3}
    J_{\Pi} \geq 0,\quad \tr_{A'B'} J_{\Pi} = \1_{AB},\quad J_{\Pi}^{T_{BB'}} \geq 0.
    \end{align}
    Combining Eqs.~\eqref{Lemmma 1.3 tmp1}, \eqref{Lemmma 1.3 tmp2}, \eqref{Lemmma 1.3 tmp3} we have the optimization
    \begin{subequations}
    \label{ED PPT tmp}
    \begin{align}
    E_{\rm D,PPT}^{(1),\ve}(\rho) = \log  \max & \ k \\ 
   \text{\rm s.t.} & \tr \big[\tr_{A'B'} J_{\Pi}\cdot (\1_{AB} \ox \Psi_k)\big] \rho_{AB}^T \geq 1-\ve,\\
    &\ J_{\Pi} \geq 0,\ \tr_{A'B'} J_{\Pi} = \1_{AB},\ J_{\Pi}^{T_{BB'}} \geq 0.
    \end{align}
    \end{subequations}
    Suppose one optimal solution in optimization~\eqref{ED PPT tmp} is given by $\widetilde J_{\Pi}$. 
    Since $\Psi_k$ is invariant under any local unitary $U_{A'}\ox \overline U_{B'}$, i.e., $(U_{A'}\ox \overline U_{B'}) \Psi_k (U_{A'}\ox \overline U_{B'})^\dagger = \Psi_k$, we can verify that $(U_{A'}\ox \overline U_{B'}) \widetilde J_{\Pi} (U_{A'}\ox \overline U_{B'})^\dagger$ is also optimal. 
    Since any convex combination of optimal solutions remains optimal, we know that 
    \begin{align}
         \int dU (U_{A'} \ox \overline{U}_{B'}) \widetilde J_{\Pi} (U_{A'} \ox \overline{U}_{B'})^\dagger
     \end{align} 
     is optimal, where $dU$ is the Haar measure. According to Schur's lemma, the result of the above integral gives an optimal solution admitting the structure of $W_{AB} \ox \Psi_k + Q_{AB} \ox (\1 - \Psi_k)$ with certain linear operators $W_{AB}$ and $Q_{AB}$. Thus without loss of generality, we can restrict our consideration of the optimal \Choi matrix in the optimization~\eqref{ED PPT tmp} as
    \begin{align}\label{choi ansatz}
    J_{\Pi} = W_{AB} \ox \Psi_k + Q_{AB} \ox (\1 - \Psi_k) .
    \end{align}
    In the following we take Eq.~\eqref{choi ansatz} into the optimization~\eqref{ED PPT tmp} and perform some further simplifications. Denote $P_+$ and $P_-$ as the symmetric and anti-symmetric projections respectively.
    From the spectral decomposition, we know that $\Psi_k^{T_{B'}} = (P_+ - P_-)/k$ and  
    \begin{align}
        J_{\Pi}^{T_{BB'}} & = W_{AB}^{T_B} \ox \Psi_k^{T_{B'}} + Q_{AB}^{T_B}\ox (\1 - \Psi_k)^{T_{B'}} \\
        & = W_{AB}^{T_B} \ox  \frac{P_+ - P_-}{k}  + Q_{AB}^{T_B}\ox \frac{(k-1)P_+ + (k+1)P_-}{k}\\
        & = \Big[ W_{AB}^{T_B}+(k-1)Q_{AB}^{T_B}\Big] \ox \frac{P_+}{k} + \Big[-W_{AB}^{T_B}+(k+1)Q_{AB}^{T_B}\Big] \ox \frac{P_-}{k}.
    \end{align}
    Since $P_+$ and $P_-$ are positive and orthogonal to each other, we have $J_{\Pi}^{T_{BB'}} \geq 0$ if and only if 
    \begin{align}
    W_{AB}^{T_B}+(k-1)Q_{AB}^{T_B} & \geq 0,\\
    -W_{AB}^{T_B}+(k+1)Q_{AB}^{T_B} & \geq 0.
    \end{align}
     Note that $\tr \big[\tr_{A'B'} J_{\Pi}\cdot (\1_{AB} \ox \Psi_k)\big] \rho_{AB}^T = \tr  W_{AB} \rho_{AB}^T$.
     We can simplify the optimization~\eqref{ED PPT tmp} as 
    \begin{subequations}
    \begin{align}
    E_{\rm D,PPT}^{(1),\ve}(\rho) = \log \max &\ k \\
   \text{\rm s.t.} & \tr  W_{AB} \rho_{AB}^T\geq 1-\ve,\, W_{AB},\ Q_{AB} \geq 0\\
    &\ W_{AB}+(k^2-1)Q_{AB} = \1_{AB},\\
    & (1-k)Q_{AB}^{T_B} \leq W_{AB}^{T_B} \leq (1+k)Q_{AB}^{T_B}.
    \end{align}\end{subequations}
    Eliminating the variable $Q_{AB}$ via the condition $ W_{AB}+(k^2-1)Q_{AB} = \1_{AB}$ and taking $M_{AB} = W_{AB}^T$, $\eta = 1/k$, we obtain the desired result.
\end{proof}

\vspace{0.2cm}
\begin{remark}
    The optimization in the above proposition is not exactly an SDP due to the nonlinear objective function. However, to compute the one-shot distillable entanglement, we can first implement the SDP:
    \begin{align}\label{eq:SDP ED1PPT}
    \min \big\{\, \eta \, \big| \, 0\le M_{AB} \le \1_{AB},\, \tr \rho_{AB}M_{AB}\ge 1-\ve,\, -\eta\1_{AB} \le M_{AB}^{T_{B}} \le  \eta\1_{AB} \, \big\}
\end{align}
and obtain the optimal value $\eta_0$. Then the one-shot distillable entanglement is given by $\log \lfloor 1/\eta_0 \rfloor$. Since the second step is trivial, we will also call the Eq.~\eqref{E 1 PPT} an SDP characterization.

 We also note that it is possible to use the SDP of distillation fidelity in \cite{Rains2001} to obtain Eq.~\eqref{eq:SDP ED1PPT}. However, for a general task, the SDP of its optimal fidelity is not sufficient to ensure that the rate can also be characterized by an SDP. A counter-example can be given by the one-shot PPT-assisted entanglement dilution task~\cite{Matthews2008}.

\end{remark}
\vspace{0.2cm}



The one-shot entanglement distillation under non-entangling operations was studied by Brand\~{a}o and Datta in~\cite{Brandao2011a}. Particularly they provided both lower and upper (non-matching) bounds of the one-shot distillable rate via the quantum hypothesis testing relative entropy. 
Based on the SDP formula in Proposition~\ref{EDPPT SDP characterization}, we are now ready to give a similar but exact characterization for distillation under PPT operations.

\begin{theorem} 
    \label{EDPPT DH relation}
    For any bipartite quantum state $\rho_{AB}$ and the error tolerance $\ve \in (0,1)$, it holds
    \begin{align}
    \label{EDPPT DH relation eq}
    E_{\DPPT}^{(1),\ve}(\rho_{AB}) = \min_{\substack{\|G^{T_B}\|_1 \leq 1\\ G = G^\dagger}} D_H^{\ve}(\rho_{AB}\|G_{AB}) - \delta,
    \end{align}
 where $\delta \in [0,1]$ is the least constant such that the r.h.s. is the logarithm of an integer.
\end{theorem}
\begin{proof}
    We first note the fact that $\delta = \log x - \log \lfloor x \rfloor \in [0,1]$ for any $x \geq 1$. Thus we can use the least constant $\delta \in [0,1]$ to adjust the r.h.s. of Eq.~\eqref{E 1 PPT} to be the logarithm of an integer without taking the floor function, i.e.,
    \begin{align}\label{SDP delta}
    E_{\DPPT}^{(1),\ve}(\rho_{AB}) = -\log \min \big\{\, \eta \, \big| \, 0\le M\le \1,\, \tr \rho M\ge 1-\ve,\, -\eta\1 \le M^{T_{B}} \le  \eta\1 \, \big\} - \delta.
\end{align}
    The main ingredient of the following proof is the norm duality between the trace norm and the operator norm. Denote the set $\cS_\rho := \left\{\,M_{AB} \,|\, 0\le M_{AB} \le \1_{AB},\, \tr \rho_{AB}M_{AB}\ge 1-\ve \,\right\}$.
    Then we have
    \begin{align}
    E_{\DPPT}^{(1),\ve}(\rho_{AB}) 
    &= -\log \min_{M \in \cS_\rho} \|M_{AB}^{T_{B}}\|_\infty - \delta\\
    & = -\log \min_{M \in \cS_\rho} \max_{\substack{\|G\|_1 \leq 1\\ G = G^\dagger}} \tr M_{AB}^{T_{B}}G_{AB} - \delta\\
    & = -\log \max_{\substack{\|G\|_1 \leq 1\\ G = G^\dagger}} \min_{M \in \cS_\rho}  \tr M_{AB}^{T_{B}}G_{AB} -\delta\\
    & = -\log \max_{\substack{\|G^{T_B}\|_1 \leq 1\\ G = G^\dagger}}  \min_{M \in \cS_\rho}  \tr M_{AB}G_{AB} - \delta\\
    & = \min_{\substack{\|G^{T_B}\|_1 \leq 1\\ G = G^\dagger}} -\log   \min_{M \in \cS_\rho}  \tr M_{AB}G_{AB} - \delta\\
    & = \min_{\substack{\|G^{T_B}\|_1 \leq 1\\ G = G^\dagger}}  D_H^{\ve}(\rho_{AB}\|G_{AB}) - \delta.
    \end{align}
    The first line follows from Eq.~\eqref{SDP delta} and $\|X\|_\infty = \min\{\,\eta\,|\, -\eta \1 \leq X \leq \eta \1\,\}$. The second line follows from the norm duality between the trace norm and the operator norm, i.e., $\|X\|_\infty = \max_{\|Y\|_1\leq 1, Y =Y^\dagger} \tr XY$. The third line uses the Sion minimax theorem \cite{Sion1958} to swap the minimization with the maximization. In the fourth line, we replace $G_{AB}$ with $G_{AB}^{T_B}$. The last line follows by definition.
\end{proof}

\vspace{0.2cm}
Compared to the result by Brand\~{a}o and Datta~\cite{Brandao2011a}, Theorem~\ref{EDPPT DH relation} gives an exact and complete characterization for one-shot entanglement distillation under PPT operations.
The proof technique of this theorem was later applied in coherence theory, where the one-shot distillable coherence under maximally incoherent operations is also completely characterized by the quantum hypothesis testing relative entropy~\cite{Regula2017c}. 
It is also worth noting that the second argument $G_{AB}$ in Eq.~\eqref{EDPPT DH relation eq} is not necessarily positive, with an explicit example presented in Appendix~\ref{c not positive}. Thus the extension to general Hermitian operators becomes crucial to obtain the exact characterization instead of non-matching bounds.

In term of the asymptotic distillable entanglement, the best known upper bound was given by the Rains bound~\cite{Rains2001,Audenaert2002}. That is, $E_{\DPPT}(\rho_{AB})\le R(\rho_{AB})$ with
\begin{align}\label{Rains bound definition}
    R(\rho_{AB}) = \min_{\substack{\|\sigma^{T_B}\|_1 \leq 1\\ \sigma_{AB} \geq 0}} D(\rho_{AB}\|\sigma_{AB}),
 \end{align}
 where the quantum relative entropy $D(\rho\|\sigma):= \tr \rho (\log \rho - \log \sigma)$ if $\supp\, \rho \subseteq \supp\, \sigma$ and $+\infty$ otherwise.
Theorem~\ref{EDPPT DH relation} can be seen as a one-shot analog of this result, since we can quickly recover the Rains bound through the quantum Stein's lemma~\cite{Hiai1991}.  

\begin{corollary}\label{corol: recover rains bound}
For any bipartite quantum state $\rho_{AB}$, it holds $E_{\DPPT}(\rho_{AB})\le R(\rho_{AB}).$
\end{corollary}
\begin{proof}
    Denote a minimizer of the Rains bound as $\sigma_{AB}$ and then $R(\rho_{AB}) = D(\rho_{AB}\|\sigma_{AB})$. According to Theorem~\ref{EDPPT DH relation}, we have 
    \begin{align}
        \lim_{\ve \to 0}\lim_{n \to \infty} \frac1n E_{\DPPT}^{(1),\ve}(\rho_{AB}^{\ox n}) \leq \lim_{\ve \to 0}\lim_{n \to \infty}  \frac1n D_H^\ve(\rho_{AB}^{\ox n}\|\sigma_{AB}^{\ox n}),
    \end{align}
    since $\sigma_{AB}^{\ox n}$ is a feasible solution. The l.h.s. gives $E_{\DPPT}(\rho_{AB})$ by definition while the r.h.s. converges to $D(\rho_{AB}\|\sigma_{AB})$ due to the quantum Stein's lemma~\cite{Hiai1991}. 
\end{proof}

\section{Non-asymptotic entanglement distillation}
\label{Non-asymptotic entanglement distillation}

In this section, we study the estimation of distillable entanglement for given $n$ copies of the resource state. We provide both lower and upper bounds that are efficiently computable for general quantum states. 

Before proceeding, we need to introduce some basic notations. 
The purified distance between two subnormalized quantum states is defined as 
$P(\rho,\sigma) := \sqrt{1-F^2(\rho,\sigma)}$ with the generalized fidelity $F(\rho,\sigma):=\|\rho^{1/2}\sigma^{1/2}\|_1 + \sqrt{(1-\tr \rho)(1-\tr \sigma)}$~\cite{Tomamichel2015b}.
Denote the $\ve$-ball around $\rho_{AB}$ as $\cB_\ve(\rho_{AB}) := \{\,\widetilde \rho_{AB} \in \cS_\leq(AB) \,|\, P(\rho_{AB},\widetilde \rho_{AB}) \leq \ve \,\}$. The smooth conditional max-entropy is defined as 
\begin{align}
H_\text{max}^\ve(A|B)_\rho := \inf_{\widetilde \rho_{AB} \in \cB_\ve(\rho_{AB})} \sup_{\sigma_B \in \cS(B)} \log F(\widetilde \rho_{AB}, \1_A \ox \sigma_B).
\end{align}

The following lemma gives the second-order expansion of the quantum hypothesis testing relative entropy and the smooth conditional max-entropy. This lemma is crucial to obtain our second-order bounds.

\begin{lemma}
    \label{second order lemma}
    For any quantum states $\rho_{AB}$ and positive operator $\sigma_{AB}$, it holds~\cite{Tomamichel2013a,Li2014a}
    \begin{align}
    & D_H^\ve\left(\rho_{AB}^{\ox n}\|\sigma_{AB}^{\ox n}\right) = 
    nD(\rho_{AB}\|\sigma_{AB})+\sqrt{nV(\rho_{AB}\|\sigma_{AB})}\,\Phi^{-1}(\ve)+O(\log n),\\
    & H_{\max}^\ve(A^n|B^n)_{\rho^{\ox n}} = - n I(A\>B)_\rho - \sqrt{n V(A\>B)_\rho}\, \Phi^{-1}(\ve^2) + O(\log n),\end{align}
    where $V(\rho\|\sigma) := \tr \rho (\log \rho - \log \sigma)^2 - D(\rho\|\sigma)^2$ is the quantum information variance, $I(A\>B)_\rho := D(\rho_{AB}\|\1_A\ox \rho_B)$ is the coherent information, $V(A\>B)_\rho  :=  V(\rho_{AB}\|\1_A\ox \rho_B)$ is the coherent information variance and 
    $\Phi^{-1}$ is the inverse of cumulative normal distribution function.
\end{lemma}


\begin{theorem}\label{second order theorem}
For any bipartite quantum state $\rho_{AB}$, the number of prepared states $n$, the error tolerance $\ve \in (0,1)$, and the operation class $\O \in \{\LLOCC,\LOCC,\SEP,\PPT\}$, it holds
\begin{align}\label{theorem: second_order}
 f(\rho,n,\ve) + O(\log n) \leq & \ E_{\text{\rm D},\O}^{(1),\ve} (\rho_{AB}^{\ox n})  \le g(\rho,n,\ve) + O(\log n),
\end{align}
where $f(\rho,n,\ve)$ and $g(\rho,n,\ve)$ are efficiently computable functions given by
\begin{align}
f(\rho,n,\ve) & = n  I(A\>B)_\rho  +  \sqrt{n V(A\>B)_\rho}\, \Phi^{-1}(\ve), \label{f}\\
g(\rho,n,\ve)& = n R(\rho_{AB}) +\sqrt{nV(\rho_{AB}\|\sigma_{AB})} \, \Phi^{-1}(\ve), \label{g}
\end{align}
and $\sigma_{AB}$ is any minimizer of the Rains bound, $\Phi^{-1}$ is the inverse of cumulative normal distribution function.
\end{theorem}
\begin{proof}
Due to the inclusion relations of the operation classes, we only need to show the upper bound for PPT operations and lower bound for 1-LOCC operations. Each second-order bound can be obtained by applying the corresponding one-shot bound to the $n$-fold tensor product state $\rho^{\ox n}$ and using the second-order expansion of the related entropies.  

For the second-order upper bound, the proof steps are similar to Corollary~\ref{corol: recover rains bound}. Denote $\sigma_{AB}$ as a minimizer of the Rains bound. We first have $E_{\rm D,PPT}^{(1),\ve}(\rho_{AB}^{\ox n}) \leq D_H^{\ve}(\rho_{AB}^{\ox n}\|\sigma_{AB}^{\ox n}) - \delta$ by taking a feasible solution $\sigma_{AB}^{\ox n}$ in Theorem~\ref{EDPPT DH relation} and $\delta \in [0,1]$. Instead of applying the quantum Stein's lemma here, we use the second-order expansion of the quantum hypothesis testing relative entropy in Lemma~\ref{second order lemma} and obtain 
\begin{align}
E_{\rm D,\PPT}^{(1),\ve}\left(\rho_{AB}^{\ox n}\right) \leq nD(\rho_{AB}\|\sigma_{AB})+\sqrt{nV(\rho_{AB}\|\sigma_{AB})}\,\Phi^{-1}(\ve)+O(\log n).
\end{align}

For the second-order lower bound, we adopt the one-shot hashing bound~\cite{Wilde2016c} that
\begin{align}
    E_{\text{D,1-LOCC}}^{(1),\ve}(\rho_{AB}) \geq -H_{\text{max}}^{\sqrt{\ve}-\eta}(A|B)_\rho + 4 \log \eta, \quad \text{with} \quad \eta \in [0,\sqrt{\ve}\big).
  \end{align} 
  For the state $\rho_{AB}^{\ox n}$, we choose $\eta = 1/\sqrt{n}$ and have the following result which holds for $n > 1/\ve$,
  \begin{align}
  E_{\text{D,1-LOCC}}^{(1),\ve}(\rho_{AB}^{\ox n}) \geq -H_{\text{max}}^{\sqrt{\ve}-1/\sqrt{n}}(A^n|B^n)_{\rho^{\ox n}} + 4 \log (1/\sqrt{n}).
  \end{align}
  Using the second-order expansion of the smooth conditional max-entropy in Lemma~\ref{second order lemma}, we have
  \begin{align}
E_{\text{D,1-LOCC}}^{(1),\ve}(\rho_{AB}^{\ox n}) & \geq n I(A\>B)_\rho + \sqrt{n V(A\>B)_\rho}\, \Phi^{-1}\big((\sqrt{\ve}-1/\sqrt{n}\,)^2\big)+ O(\log n).
  \end{align}
Note that $\Phi^{-1}$ is continuously differentiable around $\ve >0$. Thus $\Phi^{-1}\big((\sqrt{\ve}-1/\sqrt{n}\,)^2\big) = \Phi^{-1}(\ve) + O(1/\sqrt{n})$ and we have the desired result.
\end{proof}

\vspace{0.2cm}
Note that the second-order upper bound works for any minimizer $\sigma_{AB}$ of the Rains bound regardless of its uniqueness. Thus we can choose the one that gives the tightest result. Due to the higher-order term $O(\log n)$, the estimations in Theorem~\ref{second order theorem} will work better for large blocklength $n$ where the logarithmic term is negligible.

Since the Rains bound in Eq.~\eqref{Rains bound definition} is given by convex optimization, there are various methods to solve it numerically. We provide an algorithm in Appendix~\ref{Numerical estimation of Rains bound} which can be used to efficiently compute the Rains bound and output the minimizer operator. It is worth noting that our bounds are similar to the second-order bounds on quantum capacity in~\cite{Tomamichel2016}. But those bounds in~\cite{Tomamichel2016} are not easy to compute in general.

The difficulty to obtain good second-order estimations is to find suitable one-shot lower and upper bounds which lead to the same $\ve$ dependence in the term $\Phi^{-1}(\ve)$ after the second-order expansion. This is necessary to show the tightness of the second-order estimation for specific states. Our result in Theorem~\ref{EDPPT DH relation} and the one-shot lower bound in~\cite{Wilde2016c} coordinate well in this sense.
    There are other one-shot lower bounds \cite{Buscemi2013,Buscemi2010b}  which can be used to establish a second-order estimation. But they do not provide matching $\ve$ dependence with our second-order upper bound. For certain pure states, there exists a better one-shot lower bound in \cite{Buscemi2013}. But note that our bounds are already tight for general pure states up to the second-order terms (see Proposition~\ref{prop: pure state}).


\section{Examples}\label{sec:examples}

In this section, we apply our second-order bounds to estimate the non-asymptotic distillable entanglement of some important classes of states, including pure states, mixtures of Bell states, maximally correlated states and isotropic states.

\subsection*{Pure states}

\begin{proposition}\label{prop: pure state}
    For any bipartite pure state $\psi_{AB}$, the number of prepared states $n$, the error tolerance $\ve \in (0,1)$, and the operation class $\O \in \{\LLOCC,\LOCC,\SEP,\PPT\}$, it holds 
    \begin{align}
\label{tight pure}
E_{\rm D,\O}^{(1),\ve}(\psi_{AB}^{\ox n}) = n S(\rho_A) +
\sqrt{n(\tr\rho_A (\log \rho_A)^2 - S(\rho_A)^2)}\, \Phi^{-1}(\ve) + O(\log n),
\end{align}
where $S(\rho):= -\tr \rho \log \rho$ is the von Neumann entropy and $\rho_A = \tr_B \psi_{AB}$.
\end{proposition}
\begin{proof}
Since all the quantities concerned are invariant under local unitaries, we only need to consider a pure state $\psi_{AB}$ with the Schmidt decomposition $\ket{\psi_{AB}} =\sum_i \sqrt{p_i} \ket{i_Ai_B}$. Then $\rho_A = \sum_i p_i \ket{i_A}\bra{i_A}$ and $\rho_B := \tr_A \psi_{AB} = \sum_i p_i \ket{i_B}\bra{i_B}$. Let $\sigma_{AB} =\sum_i p_i \ket{i_Ai_B}\bra{i_Ai_B}$. The following equalities are straightforward by calculation,
\begin{align}
 D(\psi_{AB}\|\sigma_{AB}) & = I(A\>B)_\psi = S(\rho_A), \\
 V(\psi_{AB}\|\sigma_{AB}) & = V(A\>B)_\psi = \tr\rho_A (\log \rho_A)^2 - S(\rho_A)^2.
\end{align}
We can first check that $\sigma_{AB}$ is a feasible solution for the Rains bound and thus $R(\psi_{AB}) \leq D(\psi_{AB}\|\sigma_{AB})$. Note that $I(A\>B)_\rho \leq R(\rho_{AB})$ holds for any quantum state $\rho_{AB}$. Thus we have $D(\psi_{AB}\|\sigma_{AB}) = I(A\>B)_\rho \leq R(\rho_{AB}) \leq D(\psi_{AB}\|\sigma_{AB})$, which implies that $\sigma_{AB}$ is a minimizer of the Rains bound. Finally applying Theorem~\ref{second order theorem}, we have the desired result.
\end{proof}

\vspace{0.2cm}
The second-order estimation for pure states has been given by Datta and Leditzky \cite{Datta2015a}, which was only for LOCC operations. It is known that the asymptotic distillable entanglement of a pure state coincides with the von Neumann entropy of its reduced state under 1-LOCC, LOCC, SEP or PPT operations~\cite{lo2001concentrating}. Proposition~\ref{prop: pure state} shows that not only are the asymptotic distillable entanglement (first-order asymptotics) the same under these four sets of operations but also their convergence speeds (second-order asymptotics).

\subsection*{Mixture of Bell states}

In laboratories, we usually obtain mixed states due to the imperfection of operations and decoherence. A common case is a noise dominated by one type of Pauli error \cite{Childress2006,Campbell2007}. Without loss of generality, we consider the phase noise, which results in the mixture of Bell states 
\begin{align}
    \rho_{\text{Bell}} = p \ket{v_1}\bra{v_1} + (1-p) \ket{v_2}\bra{v_2}, \quad 0 < p < 1,
\end{align} where $\ket{v_1} = (\ket{01} + \ket{10})/\sqrt{2}$ and $\ket{v_2} = (\ket{01} - \ket{10})/\sqrt{2}$. Let $\sigma_{AB} = (\ket{v_1}\bra{v_1} + \ket{v_2}\bra{v_2})/2$. Following similar proof steps as Proposition~\ref{prop: pure state}, we can check that $\sigma_{AB}$ is a minimizer of the Rains bound for $\rho_{\text{Bell}}$. Due to Theorem~\ref{second order theorem}, we have the following result.

\begin{proposition}\label{prop: bell state}
    For given quantum state $\rho_{\text{\rm Bell}}$, the number of prepared states $n$, the error tolerance $\ve \in (0,1)$, and the operation class $\O \in \{\LLOCC,\LOCC,\SEP,\PPT\}$, it holds 
    \begin{align}
\label{second order mixed state ex1}
    E_{\rm D,\O}^{(1),\ve}(\rho_{\text{\rm Bell}}^{\ox n}) = n (1-h_2(p)) +  
    \sqrt{np(1-p)(\log \frac{1-p}{p})^2}\, \Phi^{-1}(\ve) + O(\log n),
\end{align}
where $h_2(p) := -p \log p - (1-p) \log (1-p)$ is the binary entropy.
\end{proposition}


\subsection*{Maximally correlated states}
Besides the mixture of Bell states presented above, we can also show the tightness of our second-order bounds for a broader class of states, i.e., maximally correlated states 
\begin{align}
    \rho_{\text{mc}} = \sum_{i,j= 0}^{d-1}\widetilde \rho_{ij} \ket{i_Ai_B}\bra{j_Aj_B},
\end{align}
where $\widetilde \rho_A = \sum_{i,j=0}^{d-1}\widetilde \rho_{ij} \ket{i_A}\bra{j_A}$ is a quantum state. Denote $\Delta(\cdot) = \sum_{i,j=0}^{d-1} \<i_Aj_B|\cdot|i_Aj_B\> \ket{i_Aj_B}\bra{i_Aj_B}$ as the completely dephasing channel on the bipartite systems. Let $\sigma_{AB} = \Delta(\rho_{\text{mc}})$. We can check that $\sigma_{AB}$ is a minimizer of the Rains bound for $\rho_{\text{mc}}$. Due to Theorem~\ref{second order theorem}, we have the following result.

\begin{proposition}\label{prop: mc state}
    For any maximally correlated state $\rho_{\text{\rm mc}}$, the number of prepared states $n$, the error tolerance $\ve \in (0,1)$, and the operation class $\O \in \{\LLOCC,\LOCC,\SEP,\PPT\}$, it holds 
\begin{align}
\label{second order maximally correlated state}
    E_{\rm D,\O}^{(1),\ve}(\rho_{\text{\rm mc}}^{\ox n}) = n D(\rho_{\text{\rm mc}} \,\|\Delta(\rho_{\text{\rm mc}})) + \sqrt{nV(\rho_{\text{\rm mc}}\,\|\Delta(\rho_{\text{\rm mc}}))}\, \Phi^{-1}(\ve) + O(\log n).
\end{align}
\end{proposition}

Note that if $\widetilde \rho_A$ is a pure state, the maximally correlated state $\rho_{\text{\rm mc}}$ will also reduce to a bipartite pure states. The class of maximally correlated states also contains the mixture of Bell states $\rho_{\text{Bell}}$. By direct calculations, we also have $D(\rho_{\text{\rm mc}} \,\|\Delta(\rho_{\text{\rm mc}})) = I(A\>B)_{\rho_{\text{\rm mc}}}$ and $V(\rho_{\text{\rm mc}}\,\|\Delta(\rho_{\text{\rm mc}})) = V(A\>B)_{\rho_{\text{\rm mc}}}$.

Furthermore, the coherence theory is closely related to the entanglement theory due to the one-to-one correspondence 
between $\rho_{\text{\rm mc}} = \sum_{i,j=0}^{d-1} \widetilde \rho_{ij}\ket{i_Ai_B}\bra{j_Aj_B}$ and $\widetilde \rho_A = \sum_{i,j=0}^{d-1}\widetilde\rho_{ij}\ket{i_A}\bra{j_A}$.
An interesting conjecture with a plethora of evidence is that any incoherent operation acting on a state $\widetilde \rho_A$ is equivalent to a LOCC operation acting on the associated maximally correlated state $\rho_{\text{\rm mc}}$~\cite{Winter2015a,Streltsov2016}. If this conjecture holds, Proposition~\ref{prop: mc state} will also give the second-order estimation for non-asymptotic coherence distillation.

\subsection*{Isotropic states}
Another common noise, in practice, is the so-called depolarizing noise \cite{Campbell2007,Nickerson2014}, which results in an isotropic state,
\vspace{-0.1cm}
\begin{align}
\rho_{F} = F\cdot \Psi_d + (1-F) \frac{\1 - \Psi_d}{d^2 -1}, \quad 0 \leq F \leq 1,
\end{align}
where $d$ is the local dimension of the maximally entangled state $\Psi_d$. The isotropic state is also the Choi-Jamio\l{}kowski state of the depolarizing channel $\cN(\rho) = p\rho + (1-p) \1/d$. Its $\LLOCC$ distillable entanglement is equal to the quantum capacity of the depolarizing channel \cite{Bennett1996c,Leditzky2018,leditzky2018useful}, the determination of which is still a big open problem in quantum information theory. Here we study the non-asymptotic distillable entanglement of this particular class of states. For small blocklength $n$ (e.g. $n \leq 100$), we can compute the exact distillation rate via a linear program. For large blocklength $n$ (e.g $n > 100$), we need to employ the second-order estimation in Theorem~\ref{second order theorem}.

Isotropic states possess the same symmetry as the maximally entangled states, which are invariant under any local unitary $U_A\ox \overline U_B$. Exploiting such symmetry, we can simplify the PPT-assisted distillable entanglement for the $n$-fold isotropic state as a linear program. We note that the optimal fidelity for $n$-fold isotropic states can also be simplified to a linear program, which has been studied by Rains in~\cite{Rains2001}. 
Here, we follow Rains' approach of utilizing the structure of isotropic states and focus on the distillable rate of $n$-fold isotropic states under a given infidelity tolerance.

\begin{proposition}   
For any $n$-fold isotropic state $\rho_F^{\ox n}$ with integer $n$ and the error tolerance $\ve$, its one-shot distillable entanglement under PPT operations $E_{\rm D,PPT}^{(1),\ve}(\rho_{F}^{\ox n})$ is given by
\begin{subequations}\label{ED PPT LP}
\begin{align}
  \log \ \  \max \ &\ \ \lfloor 1/\eta\rfloor \\ 
   \text{\rm s.t.}\   &\ \ 0 \leq m_i \leq 1, \ \forall \, i = 0,1,\cdots, n,\\
   &\ \sum\nolimits_{i=0}^n \binom{n}{i} F^i (1-F)^{n-i} m_i  \geq 1-\ve, \\
   &\ -\eta \leq \sum\nolimits_{i=0}^n x_{i,k} m_i \leq \eta, \ \forall \, k = 0,1,\cdots, n,
\end{align}
\end{subequations}
where the coefficients
\begin{align}
x_{i,k} = \frac{1}{d^n}\sum_{m=\max\{0,i+k-n\}}^{\min \{i,k\}} \binom{k}{m}\binom{n-k}{i-m} (-1)^{i-m} (d-1)^{k-m}(d+1)^{n-k+m-i}.
\end{align}
\end{proposition}
\begin{proof}  
The technique is very similar to the one we use in the proof of Proposition~\ref{EDPPT SDP characterization}.
Consider the $n$-fold isotropic state 
\begin{align}
\rho_F^{\ox n} = \sum_{i=0}^n f_i P_i^n(\Psi_d,\Psi_d^\perp),\ \text{with}\ f_i = F^i (\frac{1-F}{d^2-1})^{n-i}, \Psi_d^\perp = \1 - \Psi_d.
\end{align}
Here, $P_i^n(\Psi_d,\Psi_d^\perp)$ represents the sum of those $n$-fold tensor product terms with exactly $i$ copies of $\Psi_d$. For example,
$P_1^3(\Psi_d,\Psi_d^\perp) = \Psi_d^\perp \otimes \Psi_d^\perp \otimes \Psi_d + \Psi_d^\perp  \otimes \Psi_d \otimes \Psi_d^\perp + \Psi_d \otimes \Psi_d^\perp \otimes \Psi_d^\perp$.
Suppose $M$ is the optimal solution of the optimization  
\begin{align} \label{LP tmp 2}
E_{\rm D,PPT}^{(1),\ve}(\rho_{F}^{\ox n})=\log \max \Big\{\,\lfloor 1/\eta \rfloor\ \Big|\ 0\le M \le \1, \tr M \rho_{F}^{\ox n}\ge 1-\ve,
-\eta\1 \le M^{T_{B}} \le  \eta\1\,\Big\}.\end{align}
Then for any local unitary $U = \bigotimes_{i=1}^n \big(U_A^i\ox \overline{U}_B^i\big)$ where $i$ denotes the $i$-th copies of corresponding system, $U M U^\dagger$ is a also optimal solution. Convex combinations of optimal solutions are also optimal. So we can take the optimal solution $M$ to be an operator which is invariant under any local unitary $\bigotimes_{i=1}^n \big(U_A^i\ox \overline{U}_B^i\big)$. Moreover, since $\rho_{F}^{\ox n}$ is invariant under the symmetric group acting by permuting the tensor factors, we can take the optimal solution $M$ of the form $ \sum_{i=0}^n m_i P_i^n(\Psi_d,\Psi_d^\perp)$ without loss of generality.
 
Since $P_i^n(\Psi_d,\Psi_d^\perp)$ are orthogonal projections, the operator $M$ has eigenvalues $\{m_i\}_{i=0}^n$ without considering degeneracy. Next, we will need to know the  eigenvalues of $M^{T_B}$. Decomposing operators $\Psi_d^{T_B}$ and ${\Psi_d^{\perp}}^{T_B}$ into orthogonal projections, i.e.,
\begin{align}
    \Psi_d^{T_B} = \frac{1}{d}(P_+ - P_-), \quad {\Psi_d^{\perp}}^{T_B} = (1-\frac{1}{d})P_+ + (1+\frac{1}{d})P_-
\end{align} where $P_+$ and $P_-$ are symmetric and anti-symmetric projections respectively and collecting the terms with respect to $P_k^n(P_+,P_-)$, we have
\begin{align}
M^{T_B} & = \sum_{i=0}^n m_i P_i^n \Big(\Psi_d^{T_B},{\Psi_d^\perp}^{T_B}\Big)\\
& = \sum_{i=0}^n m_i \Big(\sum_{k=0}^n x_{i,k} P_k^n(P_+,P_-) \Big)\\
& = \sum_{k=0}^n \Big(\sum_{i=0}^n x_{i,k} m_i\Big)  P_k^n(P_+,P_-).
\end{align}
Since $P_k^n(P_+,P_-)$ are also orthogonal projections, $M^{T_B}$ has eigenvalues $\{t_k\}_{k=0}^n$ without considering degeneracy, where $t_k = \sum_{i=0}^n x_{i,k} m_i$.
As for the condition $\tr M \rho_{F}^{\ox n} \ge 1-\ve$, we have
\begin{align}
\tr M \rho_F^{\ox n} & = \tr \sum_{i=0}^n f_i m_i  P_i^n(\Psi_d, \Psi_d^\perp)\\
& = \sum_{i=0}^n f_i m_i \binom{n}{i} (d^2 - 1)^{n-i}\\
& = \sum_{i=0}^n \binom{n}{i} F^i (1-F)^{n-i} m_i. \label{LP tmp eq}
\end{align}
Taking Eq.~\eqref{LP tmp eq} and the eigenvalues of $M$, $M^{T_B}$ into Eq.~\eqref{LP tmp 2}, we have the desired result.
\end{proof}

\vspace{0.2cm}
This linear program can be solved \emph{exactly} via Mathematica. In Figure~\ref{fig:LP},
we plot the  one-shot distillable entanglement for the $n$-fold isotropic state $\rho_F^{\ox n}$ with $d = 3$, $F=0.9$, and error tolerance $0.001$. The blocklength $n$ ranges from $1$ to $100$. We observe that even if we were able to coherently manipulate 100 copies of the states with the broad class of PPT assistance, the maximal distillation rate still could not reach the hashing bound $I(A\>B)_{\rho_F}$ which is asymptotically achievable under $\LLOCC$ operations. This demonstrates that the asymptotic bounds cannot provide helpful estimations in the practical scenario.

\begin{figure}[H]
    \centering
    \includegraphics[width = 8cm]{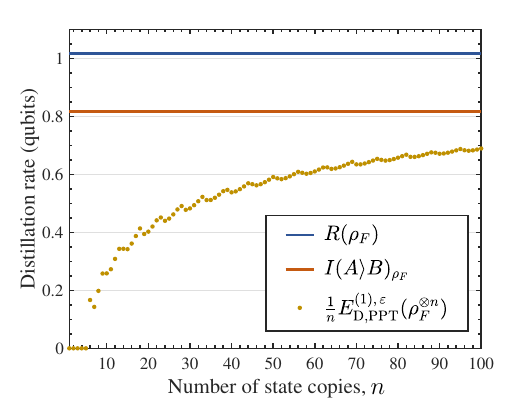}
    \caption[Linear program for distillable entanglement of isotropic state.]{The dotted line shows the exact value of distillation rate for $n$-fold isotropic state $\rho_F^{\ox n}$ with $F = 0.9$, local dimension $d = 3$. The error tolerance is taken at $\ve = 0.001$ and the blocklength $n$ ranges from $1$ to $100$. The solid line below is the  hashing bound, while the solid line above is the Rains bound.}
    \label{fig:LP}
\end{figure}

For the approximation of large blocklength distillation, we employ the second-order bounds in Theorem~\ref{second order theorem}. In Figure~\ref{fig:second_order}, we show the second-order estimation for $n$-fold isotropic state $\rho_F^{\ox n}$ with $d = 3$, $F=0.9$, and error tolerance $0.001$. In this figure we focus on the large blocklength ($n \geq 100$) regime and use a logarithmic scale for the horizontal axis. The second-order bounds in Theorem~\ref{second order theorem} are not tight, as expected, for isotropic states. But they provide a more refined estimation than the known asymptotic bounds. In Figure~\ref{fig:second_order}, the finite blocklength distillation rate lies between the two dashed lines, while the asymptotic rate lies between the two solid lines.

\begin{figure}[H]
    \centering
    \includegraphics[width = 8cm]{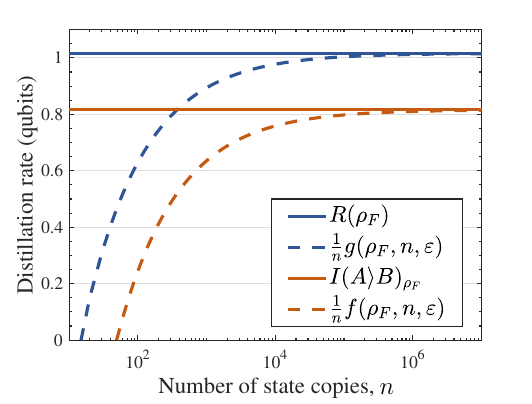}
     \caption[Second-order estimation for distillable entanglement of isotropic state.]{The two dashed lines show the second-order lower bound in Eq.~\eqref{f} and upper bound in Eq.~\eqref{g} for $n$-fold isotropic state $\rho_F^{\ox n}$ with $F = 0.9$, local dimension $d = 3$. The error tolerance is taken at $\ve = 0.001$ and the blocklength $n$ ranges from $10^2$ to $10^7$. The solid line below is the hashing bound, while the solid line above is the Rains bound.}
     \label{fig:second_order}
\end{figure}

An interesting observation is made when we present Figure~\ref{fig:LP} and Figure~\ref{fig:second_order} in a single plot. The linear program in Figure~\ref{fig:LP} is only implemented for $n$ less than $100$ due to the limited computational power. But we can use the curve fitting via least-squares  method and  construct an ansatz curve 
\begin{align}
c_1 + c_2 \frac{1}{\sqrt{n}} + c_3 \frac{\log n}{n} + c_4 \frac{1}{n},
\end{align}
which has the best fit to the series of points $\frac{1}{n}E_{\rm D,PPT}^{(1),\ve}(\rho_{F}^{\ox n})\ (1 \leq n \leq 100)$ in Figure~\ref{fig:LP}. Combining with the second-order upper bound in Figure~\ref{fig:second_order}, we get Figure~\ref{fig:LP_second_order}. It shows that for small $n$, the second-order upper bound does not give an accurate estimation since we ignore the term~$O\big(\frac{\log n}{n}\big)$. But for large $n$ ($\geq 10^2$), the fitting curve almost coincides with the second-order upper bound, demonstrating that the second-order upper bound works better for large blocklength. The convergence of the fitting curve indicates that  $E_{\rm D,PPT}(\rho_F) = R(\rho_F)$ for isotropic states $\rho_F$. It would be of great interest to find an analytical proof to this conjecture.

\begin{figure}[H]
    \centering
    \includegraphics[width = 8cm]{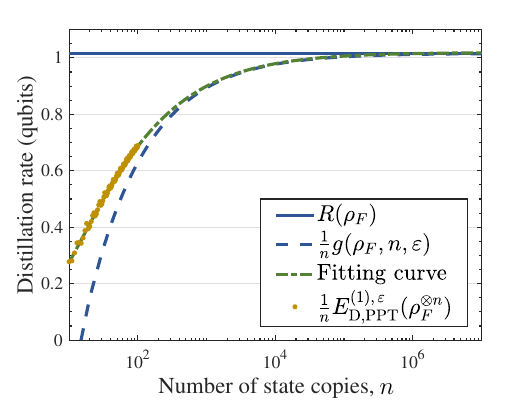}
    \caption[Curve-fitting for distillable entanglement of isotropic state.]{The dash-dotted line is the fitting curve of exact values of distillation rate for $n$-fold isotropic state $\rho_F^{\ox n}$ with $F = 0.9$, local dimension $d = 3$. The error tolerance is taken at $\ve = 0.001$. The dashed line is the second-order upper bound in Eq.~\eqref{g} and the solid line is the Rains bound.}
        \label{fig:LP_second_order}
\end{figure}

\section{Discussions}
We have provided both theoretical and numerical~\footnote{The codes for numerical calculations are available at \url{https://github.com/fangkunfred/entanglement-distillation}.} results for the entanglement distillation in the non-asymptotic regimes. Since entanglement distillation has become a central building block of quantum network proposals \cite{Dur1999,Childress2006,Gottesman2012,Nickerson2014}, our finite blocklength estimations could be applied as useful benchmarks for experimentalists to build a reliable quantum network in the future. Theoretically, we have obtained an exact characterization of the one-shot entanglement distillation under PPT operations in terms of the hypothesis testing relative entropy. This result not only leads to an improved understanding of the resource theory of entanglement, but also provides a potential approach to resolve the distillable entanglement under PPT operations or improve the Rains bound by taking other forms of feasible solution, for example, non-i.i.d. operators. 


\section*{Acknowledgments}
 We are grateful to Charles H. Bennett, Masahito Hayashi, Yinan Li, and  Andreas Winter for helpful discussions. We also thank Francesco Buscemi and Min-Hsiu Hsieh for reminding us of relevant results and references. 
 R.D., K.F. and X.W. were partly supported by the Australian Research Council, Grant No. DP120103776 and No. FT120100449.
 M.T. acknowledges an Australian Research Council Discovery Early Career Researcher Award, project No. DE160100821.

\bibliographystyle{IEEEtran}
\bibliography{TIT-Bib-DOI}

\begin{appendices}
\section{An example for Theorem~\ref{EDPPT DH relation}}
\label{c not positive}

In this section, we give an explicit example to show that the optimal solution in the optimization $\min_{\|G^{T_B}\|_1 \leq 1, G = G^\dagger} D_H^{\ve}(\rho_{AB}\|G_{AB})$ is not taken at any positive operator $G_{AB}$. Specifically, we show a strict difference between the following two optimizations:
\begin{align}
    \text{OPT1} = \min_{\substack{\|G^{T_B}\|_1 \leq 1\\ G = G^\dagger}} D_H^{\ve}(\rho_{AB}\|G_{AB}), \quad \text{OPT2} = \min_{\substack{\|G^{T_B}\|_1 \leq 1\\ G \geq 0}} D_H^{\ve}(\rho_{AB}\|G_{AB}).
\end{align} 
Note that the dual SDP of the quantum hypothesis testing relative entropy is given by 
\begin{align}
    D_H^\ve(\rho_{AB}\|G_{AB}) = -\log \max \big\{\, \tr X + t(1-\ve) \,\big| \, G-X-t\rho \geq 0,\, X \leq 0,\, t \geq 0\, \big\}.
\end{align}
Thus we have the following SDPs:
\begin{align}
    \text{OPT1} & = -\log \max \big\{ \tr X + t(1-\ve) \,\big| \, G-X-t\rho \geq 0,\, X \leq 0,\, t \geq 0,\, \|G^{T_B}\|_1 \leq 1,\, G = G^\dagger \big\},\\
    \text{OPT2} & = -\log \max \big\{ \tr X + t(1-\ve) \,\big| \, G-X-t\rho \geq 0,\, X \leq 0,\, t \geq 0,\, \|G^{T_B}\|_1 \leq 1,\, G \geq 0 \big\}.
\end{align}

We implement these two SDPs for the quantum state $\rho_\theta = \frac{3}{4}\ket{\varphi_1}\bra{\varphi_1} + \frac{1}{4}\ket{\varphi_2}\bra{\varphi_2}$ with $\ket{\varphi_1} = \cos \theta \ket{00} + \sin \theta \ket{11}$ and $\ket{\varphi_2} = \ket{10}$. The difference between OPT1 and OPT2 is  shown in Figure \ref{C_Not_Positive}. The numerics is run via the solver SDPT3 which can be solved to a very high (near-machine) precision. The maximal gap in the plot is approximately $3.4\times 10^{-2}$.

\begin{figure}[H]
    \centering
    \includegraphics[width=8cm]{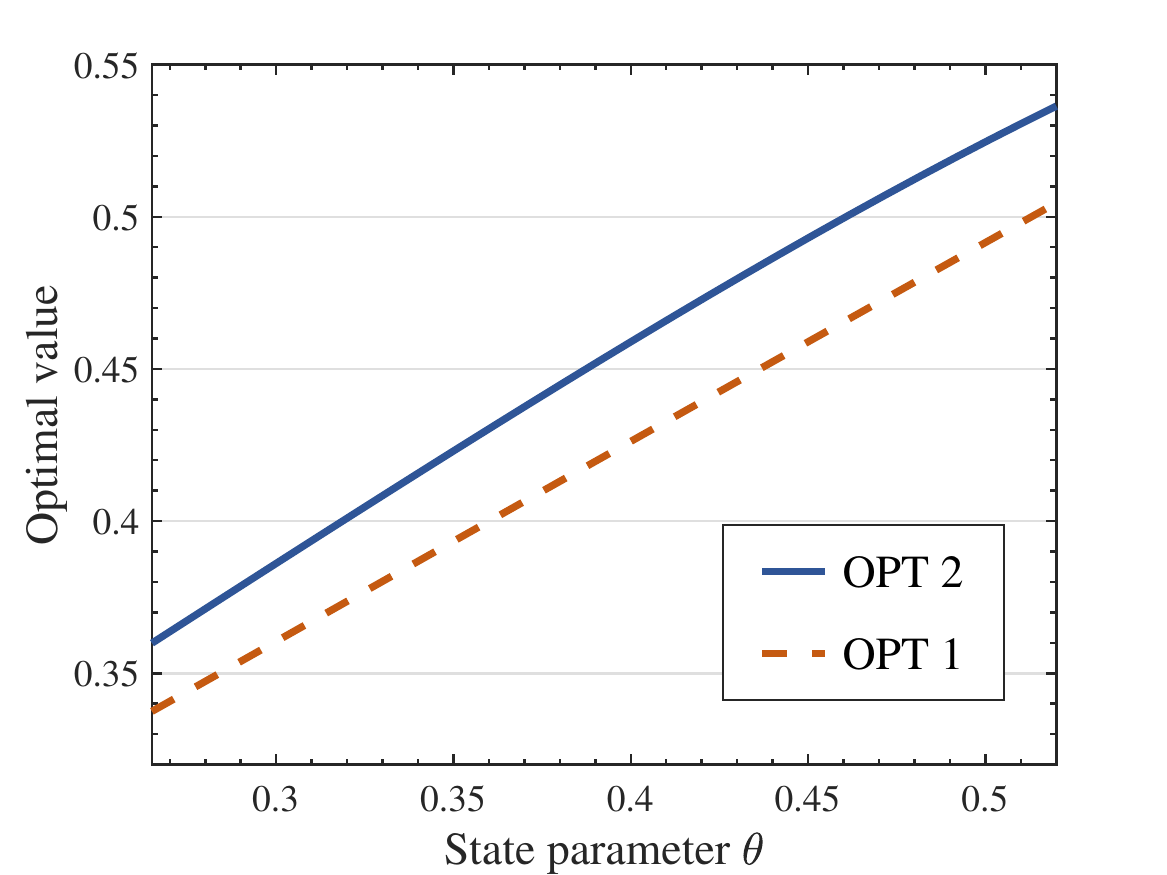}
    \caption{This figure demonstrates the difference of optimal value in OPT1 and OPT2 with respect to the state $\rho_\theta$. The solid line depicts the optimal value of OPT1 while the dashed line depicts the optimal value of OPT2. The parameter $\theta$ ranges from $\pi/12$ to $\pi/6$ and error tolerance is taken at $\ve = 1-\sqrt{3}/2$.}
    \label{C_Not_Positive}
\end{figure}

\section{Numerical estimation of Rains bound}
\label{Numerical estimation of Rains bound}


\vspace{0.5cm}

In this section, we provide an algorithm to numerically compute the Rains bound with high accuracy. In particular, the calculation of upper and lower bounds of the Rains bound can have near-machine precision while the final result of Rains bound itself is within error tolerance $10^{-6}$ by default. This algorithm closely follows the approach in~\cite{Zinchenko2010,Girard2015} which intends to compute the PPT-relative entropy of entanglement.

Note that the only difference between the Rains bound and the PPT-relative entropy of entanglement is the feasible set. Due to the similarity between these two quantities, we can have a similar algorithm for the Rains bound. For the sake of completeness, we will restate the main idea of this algorithm and clarify that our adjustment will work to compute the Rains bound. In the following discussion, we will consider the natural logarithm, denoted as $\ln$, for convenience. 

The key idea for this algorithm is based on the cutting-plane method combined with semidefinite programming. Clearly, calculating the Rains bound is equivalent to the optimization problem 
\begin{align}
\label{rains equi}
\min_{\sigma_{AB} \in \Rset}( - \tr \rho_{AB} \ln \sigma_{AB}), \quad \text{with} \quad \Rset = \Big\{\,\sigma_{AB} \geq 0 \, \Big| \, \|\sigma_{AB}^{T_B}\|_1 \leq 1\,\Big\}.
\end{align} 
If we relax the minimization over all quantum states, the optimal solution is taken at $\sigma = \rho$. Thus $-\tr\rho \ln \rho$ provides a trivial lower bound on (\ref{rains equi}).
Since the objective function is convex with respect to $\sigma$ over the Rains set ($\Rset$), its epigraph is supported by tangent hyperplanes at every interior point $\sigma^{(i)} \in \text{int}\, \Rset$. Thus we can construct a successively refined sequence of approximations to the epigraph of the objective function restricted to the interior of the Rains set. 

Specifically, for an arbitrary positive definite operator $X$, we have a spectral decomposition $X = U_X \text{diag}(\lambda_X)U_X^\dagger$ with unitary matrix $U_X$ and diagonal matrix $\text{diag}(\lambda_X)$ formed by the eigenvalues~$\lambda_X$. Then we have the first-order expansion 
\begin{equation}
    \ln (X+\Delta) = \ln X + U_X \left[D(\lambda_X)\circ U_X^\dagger \Delta U_X \right]U_X^\dagger + O(\|\Delta\|^2),
\end{equation}
where $D(\lambda)$ is the Hermitian matrix given by 
\begin{equation}
    D(\lambda)_{i,j} = \begin{cases}
    (\ln \lambda_i - \ln \lambda_j)/(\lambda_i - \lambda_j) , & \lambda_i \neq \lambda_j,\\
    1/\lambda_i, & \lambda_i = \lambda_j.
    \end{cases}
\end{equation}

For any given set of feasible points $\{\sigma^{(i)}\}_{i=0}^N \subset \text{int}\, \Rset$, we have spectral decompositions $\sigma^{(i)} = U_{(i)} \text{diag}(\lambda^{(i)})U_{(i)}^\dagger$. Then $\text{epi} (-\tr \rho \ln \sigma)|_{\text{int}\, \Rset}$ is a subset of all $(\sigma,t) \in \text{int}\,\Rset \times \RR$ satisfying 
\begin{align}
\label{algorithm ineq}
 -\tr \rho \left\{\ln \sigma^{(i)} + U_{(i)} \left[D(\lambda^{(i)})\circ U_{(i)}^\dagger (\sigma - \sigma^{(i)})U_{(i)} \right]U_{(i)}^\dagger \right\} \leq t,\ \text{for}\ i = 0,\cdots,N.\end{align}
Equivalently, we can introduce slack variable $s_i$ on the l.h.s. of Eq. (\ref{algorithm ineq}) and have 
\begin{equation}
    \tr E^{(i)} \sigma + t - s_i = -\tr \rho \ln \sigma^{(i)} + \tr E^{(i)} \sigma^{(i)},\, s_i \geq 0,\ \text{for}\ i = 0,\cdots,N,
\end{equation}
where $E^{(i)} = U_{(i)} \left[D(\lambda^{(i)})\circ U_{(i)}^\dagger \rho U_{(i)} \right]U_{(i)}^\dagger$.
So the optimal value of optimization problem \begin{align}\min \Big\{\,t \,\Big|\, \tr E^{(i)} \sigma + t - s_i = -\tr \rho \ln \sigma^{(i)} + \tr E^{(i)} \sigma^{(i)},\, s_i \geq 0,\, i = 0,\cdots,N,\, \sigma \in \Rset \,\Big\}\end{align}
provides a lower bound on (\ref{rains equi}). For any feasible point $\sigma^* \in \Rset$, $-\tr \rho \ln \sigma^*$ provides an upper bound on (\ref{rains equi}). For each iteration of the algorithm, we add a interior point $\sigma^{(N+1)}$ of the Rains set to the set $\left\{\sigma^{(i)}\right\}_{i=0}^N$, which may lead to a tighter lower bound and update the feasible point $\sigma^*$ if $\sigma^{(N+1)}$ provides a tighter upper bound. We use the variables $\overline{R}$  and  $\underline{R}$ to store the upper and lower bounds. Since $\underline{R}$ and $\overline{R}$  are nondecreasing and nonincreasing, respectively, at each iteration, we can terminate the algorithm when $\underline{R}$ and $\overline{R}$ are close enough, for example, less than given tolerance $\ve$.
The full algorithm is presented in Algorithm \ref{rains bound algorithm}.

\begin{algorithm}
\caption{Rains bound algorithm}\label{rains bound algorithm}
\begin{algorithmic}[1]
\State \textbf{Input:}  bipartite state $\rho_{AB}$ and dimensions of subsystem $d_A$, $d_B$
\State \textbf{Output:} Upper bound $\overline{R}$, lower bound $\underline{R}$

\If{$\rho \in \Rset$}
\State \textbf{return}  $\underline{R} = \overline{R} = 0$
\Else
\State \textbf{initialize} $\ve = 10^{-6}$, $N=0$, $\sigma^* = \sigma^{(0)} = \1_{AB}/(d_Ad_B)$, $\underline{R} = -\tr \rho\ln \rho$, $\overline{R} = -\tr \rho \ln \sigma^*$
\While{$\overline R - \underline R \geq \ve$}
\State \label{loop} \textbf{solve} $\min \left\{t \,| \tr E^{(i)} \sigma + t - s_i = -\tr \rho \ln \sigma^{(i)} + \tr E^{(i)} \sigma^{(i)}, s_i \geq 0,\, i = 0,\cdots,N, t \geq \underline{R}, \sigma \in \Rset \right\}$

\State store optimal solution $(\underline{t},\underline{\sigma})$ and \text{update lower bound} $\underline{R} = \underline{t}$
\If{\text{the gap between upper and lower bound is within given tolerance,} $\overline{R} - \underline{R} \leq \ve$}
\State \textbf{return} $\underline{R}$, $\overline{R}$ 
\Else
\State \text{add one more point} \label{update tangent point} $\sigma^{(N+1)}$, and \textbf{set} $N=N+1$
\If{$-\tr \rho \ln \sigma^{(N)} \leq -\tr \rho \ln \sigma^*$}
\State \textbf{update} feasible point $\sigma^* = \sigma^{(N)}$, and upper bound  $\overline{R} = -\tr \rho \ln \sigma^*$
\EndIf
\EndIf
\EndWhile
\EndIf 
\end{algorithmic}
\end{algorithm}
 
The following lemma ensures that $\sigma \in \Rset$ can be expressed as semidefinite conditions.

\begin{lemma}
\label{PPT semidefinite condition}
  The condition $\sigma \in \Rset$ holds if and only if $\sigma \geq 0$ and there exist operators $\sigma_+$, $\sigma_- \geq 0$ such that $\sigma^{T_B} = \sigma_+ - \sigma_-$ and $\tr(\sigma_+ + \sigma_-) \leq 1$.
\end{lemma}
\begin{proof}
  If $\sigma \in \Rset$, then $\sigma \geq 0$. Use the spectral decomposition $\sigma^{T_B} = \sigma_+ - \sigma_-$, where $\sigma_+ $ and $\sigma_-$ are positive operator with orthogonal support. Then $ \left|\sigma^{T_B} \right| = \sigma_+ + \sigma_-$ and $\tr(\sigma_+ + \sigma_-) = \|\sigma^{T_B} \|_1 \leq 1$. On the other hand, if there exist positive operators $\sigma_+$ and $\sigma_-$ such that $\sigma^{T_B} = \sigma_+ - \sigma_-$ and $\tr(\sigma_+ + \sigma_-) \leq 1$, then $\|\sigma^{T_B}\|_1 = \|\sigma_+ - \sigma_-\|_1 \leq \|\sigma_+\|_1 + \|\sigma_-\|_1 = \tr(\sigma_+ + \sigma_-) \leq 1$. Thus $\sigma \in \Rset$.
\end{proof}

\vspace{0.3cm}
For given $\left\{\sigma^{(i)}\right\}_{i=0}^N$, Step \ref{loop} in Algorithm \ref{rains bound algorithm} is an SDP which can be explicitly written as
\begin{equation}\begin{split}
\label{algorithm SDP}
 \min \ & t \\ 
   \text{s.t.}\ & \tr E^{(i)} \sigma + t - s_i = - \tr \rho  \ln \sigma^{(i)} + \tr E^{(i)} \sigma^{(i)},\ i = 0, \cdots, N,\\
   &  t \geq \underline{R}, \ s_i \geq 0,\ i = 0, \cdots, N,\\
   &  \sigma, \sigma_+, \sigma_- \geq 0,\ \sigma^{T_B} = \sigma_+ - \sigma_-,\ \tr(\sigma_+ + \sigma_-) \leq 1.
\end{split}\end{equation}
As for Step \ref{update tangent point}, variable $\sigma^{(N+1)}$ can be given by \begin{align}
\label{next point}
\sigma^{(N+1)} = \text{arg} \min \big\{\,-\tr \rho \ln \sigma \,\big|\, \sigma = \alpha Z + (1-\alpha)\underline{\sigma}, \alpha \in [0,1]\,\big\},\end{align}
where $Z$ is some fixed reference point. 
This one-dimensional minimization can be efficiently performed using the standard derivative-based bisection scheme \cite{Zinchenko2010}.

As a by-product, the above algorithm can be used to check the nonadditivity of the Rains bound, which has been recently proved in Ref.~\cite{Wang2016c}. We also consider the states $\rho_r$ defined in Ref.~\cite{Wang2016c}. 
Denote $\underline{R}_1$ the lower bound computed by our algorithm for $R(\rho_r)$ and $\overline{R}_2$ the upper bound computed by our algorithm for $R(\rho_r^{\ox 2})$. 
In Figure \ref{rains nonadd}, we can clearly observe that there is a strict gap between $\overline{R}_2$ and $2 \underline{R}_1$, which implies $R(\rho_r^{\ox 2}) \leq \overline{R}_2 < 2 \underline{R}_1 \leq 2 R(\rho_r)$. 
Note that $\underline{R}_1$ and $\overline{R}_2$ only depend on the SDPs in Eqs.~\eqref{algorithm SDP} and~\eqref{next point}, both of which can be solved to a very high (near-machine) precision, while the maximal gap in the plot is approximately $10^{-2}$. Thus our algorithm provides direct numerical evidence (not involving any other entanglement measures) for the nonadditivity of the Rains bound.
\begin{figure}[H]
    \centering
    \includegraphics[width=8cm]{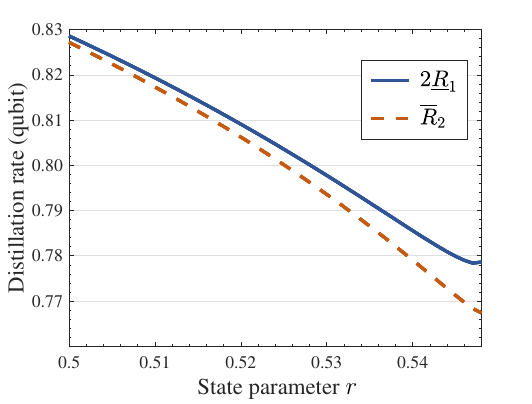}
    \caption{This figure demonstrates the difference between the lower bound $2\underline{R}_1$ on $2R(\rho_r)$ and the upper bound $\overline{R}_2$ on $R(\rho_r^{\ox 2})$. The solid line depicts $2\underline{R}_1$ while the dashed line depicts $\overline{R}_2$.}
    \label{rains nonadd}
\end{figure}

\vspace{0.2cm}
\begin{remark}
After the completion of this work, we notice that there is another approach to efficiently calculating the Rains bound in Refs.~\cite{Fawzi2017a,Fawzi2017}. In these works, the authors make use of rational (Pad\'{e}) approximations of the (matrix) logarithm function and then transform the rational functions to SDPs. Without the successive refinement, their algorithm can be much faster with relatively high accuracy. However, our algorithm is efficient enough in low-dimensional cases.
\end{remark}

\end{appendices}

\end{document}

%% file: pretex.tex
\usepackage[titletoc,title]{appendix}
\usepackage{graphicx,epic,eepic,epsfig,amsmath,latexsym,amssymb,verbatim,color}
\usepackage{dsfont}
\usepackage{float}
\usepackage{tikz}
\usepackage{times}

\usepackage{algorithm}
\usepackage[noend]{algpseudocode}

\usepackage{hyperref}
\hypersetup{colorlinks=true,citecolor=blue,linkcolor=blue,filecolor=blue,urlcolor=blue,breaklinks=true}

\usepackage{url}
\usepackage{theorem}
\newtheorem{definition}{Definition}
\newtheorem{proposition}[definition]{Proposition}
\newtheorem{lemma}[definition]{Lemma}

\newtheorem{theorem}[definition]{Theorem}
\newtheorem{corollary}[definition]{Corollary}

\def\squareforqed{\hbox{\rlap{$\sqcap$}$\sqcup$}}
\def\qed{\ifmmode\squareforqed\else{\unskip\nobreak\hfil
\penalty50\hskip1em\null\nobreak\hfil\squareforqed
\parfillskip=0pt\finalhyphendemerits=0\endgraf}\fi}
\def\endenv{\ifmmode\;\else{\unskip\nobreak\hfil
\penalty50\hskip1em\null\nobreak\hfil\;
\parfillskip=0pt\finalhyphendemerits=0\endgraf}\fi}
\newenvironment{proof}{\noindent \textbf{{Proof~} }}{\hfill $\blacksquare$}
\newenvironment{remark}{\noindent \textbf{{Remark~}}}{}

\mathchardef\ordinarycolon\mathcode`\:
\mathcode`\:=\string"8000
\def\vcentcolon{\mathrel{\mathop\ordinarycolon}}
\begingroup \catcode`\:=\active
  \lowercase{\endgroup
  \let :\vcentcolon
  }

\makeatletter
\def\resetMathstrut@{%
	\setbox\z@\hbox{%
		\mathchardef\@tempa\mathcode`\[\relax
		\def\@tempb##1"##2##3{\the\textfont"##3\char"}%
		\expandafter\@tempb\meaning\@tempa \relax
	}%
	\ht\Mathstrutbox@\ht\z@ \dp\Mathstrutbox@\dp\z@}
\makeatother
\begingroup
\catcode`(\active \xdef({\left\string(}
\catcode`)\active \xdef){\right\string)}
\endgroup
\mathcode`(="8000 \mathcode`)="8000

\newcommand{\nc}{\newcommand}
\nc{\rnc}{\renewcommand}
\nc{\beg}{\begin{equation}}
\nc{\eeq}{{\end{equation}}}
\nc{\beqa}{\begin{eqnarray}}
\nc{\eeqa}{\end{eqnarray}}
\nc{\lbar}[1]{\overline{#1}}
\nc{\bra}[1]{\langle#1|}
\nc{\ket}[1]{|#1\rangle}
\nc{\ketbra}[2]{|#1\rangle\!\langle#2|}
\nc{\braket}[2]{\langle#1|#2\rangle}

\nc{\proj}[1]{| #1\rangle\!\langle #1 |}
\nc{\avg}[1]{\langle#1\rangle}
\nc{\Rank}{\operatorname{Rank}}
\nc{\smfrac}[2]{\mbox{$\frac{#1}{#2}$}}
\nc{\tr}{\operatorname{Tr}}
\nc{\ox}{\otimes}
\nc{\dg}{\dagger}
\nc{\dn}{\downarrow}
\nc{\cA}{{\cal A}}
\nc{\cB}{{\cal B}}
\nc{\cC}{{\cal C}}
\nc{\cD}{{\cal D}}
\nc{\cE}{{\cal E}}
\nc{\cF}{{\cal F}}
\nc{\cG}{{\cal G}}
\nc{\cH}{{\cal H}}
\nc{\cI}{{\cal I}}
\nc{\cJ}{{\cal J}}
\nc{\cK}{{\cal K}}
\nc{\cL}{{\cal L}}
\nc{\cM}{{\cal M}}
\nc{\cN}{{\cal N}}
\nc{\cO}{{\cal O}}
\nc{\cP}{{\cal P}}
\nc{\cQ}{{\cal Q}}
\nc{\cR}{{\cal R}}
\nc{\cS}{{\cal S}}
\nc{\cT}{{\cal T}}
\nc{\cV}{{\cal V}}
\nc{\cX}{{\cal X}}
\nc{\cY}{{\cal Y}}
\nc{\cZ}{{\cal Z}}
\nc{\cW}{{\cal W}}
\nc{\csupp}{{\operatorname{csupp}}}
\nc{\qsupp}{{\operatorname{qsupp}}}
\nc{\var}{{\operatorname{var}}}
\nc{\rar}{\rightarrow}
\nc{\lrar}{\longrightarrow}
\nc{\polylog}{{\operatorname{polylog}}}
\nc{\1}{{\mathds{1}}}
\nc{\wt}{{\operatorname{wt}}}
\nc{\av}[1]{{\left\langle {#1} \right\rangle}}
\nc{\supp}{{\operatorname{supp}}}

\def\ve{\varepsilon}

\def\O{\Omega}

\nc{\RR}{{{\mathbb R}}}
\nc{\CC}{{{\mathbb C}}}
\nc{\FF}{{{\mathbb F}}}
\nc{\NN}{{{\mathbb N}}}
\nc{\ZZ}{{{\mathbb Z}}}
\nc{\PP}{{{\mathbb P}}}
\nc{\QQ}{{{\mathbb Q}}}
\nc{\UU}{{{\mathbb U}}}
\nc{\EE}{{{\mathbb E}}}
\nc{\id}{{\operatorname{id}}}

\nc{\CHSH}{{\operatorname{CHSH}}}

\nc{\be}{\begin{equation}}
\nc{\ee}{{\end{equation}}}
\nc{\bea}{\begin{eqnarray}}
\nc{\eea}{\end{eqnarray}}
\nc{\<}{\langle}
\rnc{\>}{\rangle}
\nc{\Hom}[2]{\mbox{Hom}(\CC^{#1},\CC^{#2})}
\nc{\rU}{\mbox{U}}

\nc{\ob}[1]{#1}

\nc{\SEP}{{\text{\rm SEP}}}
\nc{\NS}{{\text{NS}}}
\nc{\LOCC}{{\text{\rm LOCC}}}
\nc{\LLOCC}{{\text{\rm 1-LOCC}}}
\nc{\PPT}{{\text{\rm PPT}}}
\nc{\Rset}{{\text{\rm PPT}'}}
\nc{\EXT}{{\text{EXT}}}
\nc{\Sym}{{\operatorname{Sym}}}
\nc{\DLLOCC}{\text{\rm D},\text{\rm 1-LOCC}}
\nc{\DLOCC}{\text{\rm D},\text{\rm LOCC}}
\nc{\DSEP}{\text{\rm D},\text{\rm SEP}}
\nc{\DPPT}{\text{\rm D},\text{\rm PPT}}
\nc{\Choi}{\text{Choi-Jamio\l{}kowski }}

\nc{\ERLO}{{E_{\text{r,LO}}}}
\nc{\ERLOCC}{{E_{\text{r,LOCC}}}}
\nc{\ERPPT}{{E_{\text{r,PPT}}}}
\nc{\ERLOCCinfty}{{E^{\infty}_{\text{r,LOCC}}}}
\nc{\Aram}{{\operatorname{\sf A}}}